\documentclass{article}

\input{qutrit.sty}

\title{\bf Canonical forms for single-qutrit Clifford+$T$ operators}

\author{Andrew N.\ Glaudell,$^{1,2}$ Neil J.\ Ross,$^{2,3}$ and Jacob
  M.\ Taylor$^{1,2,4}$ \\[2pt]
  \small $^{1}$ Joint Quantum Institute, University of Maryland,
  College Park, MD, USA \\
  \small $^{2}$ Institute for Advanced Computer Studies and Joint
  Center for Quantum Information and Computer Science, \\
  \small University of Maryland, College Park, MD, USA \\
  \small $^{3}$ Department of Mathematics and Statistics, Dalhousie
  University, Halifax, NS, Canada \\
  \small $^{4}$ National Institute of Standards and Technology,
  Gaithersburg, MD, USA}

\date{}

\begin{document}

\maketitle

\begin{abstract}
  We introduce canonical forms for single qutrit Clifford+$T$ circuits
  and prove that every single-qutrit Clifford+$T$ operator admits a
  unique such canonical form. We show that our canonical forms are
  T-optimal in the sense that among all the single-qutrit Clifford+T
  circuits implementing a given operator our canonical form uses the
  least number of $T$ gates. Finally, we provide an algorithm which
  inputs the description of an operator (as a matrix or a circuit) and
  constructs the canonical form for this operator. The algorithm runs
  in time linear in the number of $T$ gates. Our results provide a
  higher-dimensional generalization of prior work by Matsumoto and
  Amano who introduced similar canonical forms for single-qubit
  Clifford+T circuits.
\end{abstract}

\section{Introduction}
\label{sec:intro}

Over the past five years, algebraic and number-theoretic methods have
rejuvenated the field of quantum compiling. The use of these new
mathematical techniques led to the discovery of algorithms superseding
the famed Solovay-Kitaev algorithm \cite{Dawson-Nielsen,KSV2002} for
the approximation of single-qubit unitary operators. The first such
improved algorithms were defined for the \emph{Clifford+$T$} gate set
\cite{KMM-practical,KMM-approx,RS16,S14} but these methods were later
generalized to various single-qubit gate sets
\cite{BBG2014,BGS2013,KBRY2015,KBS2013,vsynth}. Until very recently
and despite the existence of these successful methods in qubit quantum
compiling, the Solovay-Kitaev algorithm remained the standard method
in higher dimensions. However, advances in anyonic quantum computation
\cite{BCKZ15}, the discovery of protocols for higher-dimensional magic
state distillation \cite{WCAB15}, and the emergence of novel means of
error correction using higher dimensional Hilbert spaces
\cite{PSSMBC,SSCOSADMA,NCQEC,HEBQEC} have drawn the attention of the
community to qutrit quantum compiling \cite{B16,BCRS15,BCKZ15,BRS16}.

The single-qutrit Clifford+$T$ gate set, sometimes also referred to as
the \emph{supermetaplectic} gate set \cite{BCRS15}, consists of the
single-qutrit Clifford gates together with a three-dimensional
analogue of the single-qubit $T$ gate. The qutrit version of the $T$
gate was independently introduced in \cite{HV12} and \cite{WCAB15} and
shares many properties with its qubit counterpart. Most importantly,
it can be fault-tolerantly implemented via magic state distillation
\cite{WCAB15}.

In this paper, we introduce canonical forms for single-qutrit
Clifford+$T$ circuits inspired by prior work on single-qubit
Clifford+$T$ circuits \cite{FGKM15,ma-remarks,KMM,MA08}. We prove that
every single-qutrit Clifford+$T$ operator admits a canonical form and
give a linear-time algorithm to convert an arbitrary Clifford+$T$
circuit to canonical form. Finally, we show that distinct canonical
forms represent distinct operators. We establish the uniqueness of
canonical form representation by giving an algorithm which inputs the
matrix of a Clifford+$T$ operator and deterministically construct a
canonical form circuit for it. This uniqueness property implies that
our canonical forms are $T$-optimal: among all the single-qutrit
Clifford+T circuits implementing a given operator our canonical form
uses the least number of $T$ gates. This $T$-optimality is desirable
in light of the high cost associated with fault-tolerantly
implementing $T$ gates.

The paper is organized as follows. We introduce the Clifford+$T$
operators in Section~\ref{sec:ct}. We define canonical forms and prove
that every Clifford+$T$ operator admits a canonical form in
Section~\ref{sec:canonical}. Finally, we prove the uniqueness of
canonical form representations in Section~\ref{sec:uniqueness}.

\paragraph{Related work:} After the present work was completed, it was
brought to our attention that S. Prakash, A. Jain, B. Kapur, and
S. Seth independently established similar results in \cite{PJKS}.

\section{The group of single-qutrit Clifford+\texorpdfstring{$T$}{T}
  operators}
\label{sec:ct}

In what follows, $\Z$ denotes the set of integers, $\N$ the set of
nonnegative integers, and $\Z_n$ the set of integers modulo $n$. We
recall the definition of the qutrit \emph{Pauli} and \emph{Clifford}
operators \cite{G98, K96a, K96b}. Let $\omega = e^{2\pi i /3}$ and
$\zeta= e^{2\pi i / 9}$ be primitive third and ninth roots of unity so
that $\omega^3 = 1$ and $\zeta^9 = 1$. The single-qutrit \emph{Pauli
  group $\Pauli$} is generated by
\[
X := \begin{bmatrix} 0 & 0 & 1 \\ 1 & 0 & 0 \\ 0 & 1 & 0 \end{bmatrix}
\quad \mbox{and} \quad Z := \begin{bmatrix} 1 & 0 & 0 \\ 0 & \omega &
  0 \\ 0 & 0 & \omega^2 \end{bmatrix}.
\]
The Pauli group has 27 elements but only 9 if global phases can be
ignored. The single-qutrit \emph{Clifford group $\Clifford$} is the
normalizer of $\Pauli$ in the special unitary group of order 3. The
Clifford group is generated by the single-qutrit \emph{Hadamard} gate
$H$ and the single-qutrit \emph{phase} gate $S$
\[
H := \frac{1}{\sqrt{-3}}\begin{bmatrix} 1 & 1 & 1 \\ 1 & \omega &
  \omega^2 \\ 1 & \omega^2 & \omega \end{bmatrix} \quad \mbox{and}
\quad S := \zeta^8\begin{bmatrix} 1 & 0 & 0 \\ 0 & 1 & 0 \\ 0 & 0 &
\omega \end{bmatrix}
\]
where $\sqrt{-3} = i \sqrt{3}$. The above definitions of $H$ and $S$
differ from the standard ones by a global phase. As a result, both $H$
and $S$ have determinant 1. The Clifford group has 648 elements but
only 216 up to global phases. The single-qutrit \emph{$T$} gate,
introduced in \cite{HV12} and \cite{WCAB15}, is defined as
\[
T := \begin{bmatrix} 1 & 0 & 0 \\ 0 & \zeta & 0 \\ 0 & 0 &
  \zeta^8 \end{bmatrix}.
\]
The operators $H$, $S$, and $T$ form a universal gate set and generate
the single-qutrit \emph{Clifford+$T$ group}, which we sometimes denote
by $\Clifford$+$T$.

If $U$ is Clifford+$T$ operator, a \emph{circuit} for $U$ is a word
$W_1 \ldots W_k$ in the generators such that the product of the
generators appearing in $W$ is equal to $U$, i.e., $W_1 \cdot \ldots
\cdot W_k=U$. If $U$ is a Clifford+$T$ circuit, the \emph{$T$-count}
of $U$ is the number of occurences of $T$ in $U$.

\section{Canonical forms}
\label{sec:canonical}

We define a three-dimensional analogue of the canonical forms
introduced in \cite{MA08} for single qubit Clifford+$T$ circuits and
we prove that every single-qutrit Clifford+$T$ operator can be
represented by a canonical form. Our presentation follows
\cite{ma-remarks}.

\begin{definition}
  \label{def:canonical}
  A \emph{canonical form} is a Clifford+$T$ circuit of the form
  \begin{align}	\label{eq:canonicaldef}
    (\epsilon\,|\,T\,|\,H^2 T)(HT\,|\,H^3 T\,|\,S H T\,|\,S H^3
    T\,|\,S^2 H T\,|\,S^2 H^3 T)^*\Clifford.
  \end{align}
\end{definition}

Here, following \cite{ma-remarks}, we use the language of
\emph{regular expressions} to define canonical forms. In
\cref{eq:canonicaldef}, $\epsilon$ denotes the empty word and
$\Clifford$ denotes any of one of the 648 Clifford
operators. \cref{eq:canonicaldef} therefore states that a canonical
form (read from left to right) consists of an optional occurence of
$T$ or $H^2T$, any number of \emph{syllables} chosen from the set $\{H
T, H^3 T, S H T, S H^3 T, S^2 H T, S^2 H^3 T\}$, and a final Clifford
operator.

\begin{definition}
  Let $\Ss$ be the 81-element subgroup of $\Clifford$ group generated
  by $S$ and $X$, $\Mm$ be the two element subgroup of $\Clifford$
  generated by $H^2$, and $\Ll$ and $\Ll'$ be the following sets of
  Clifford operators:
  \[
  \Ll=\{\Id,H,S H, S^2 H\} ~~\mbox{ and }~~ \Ll'=\{H,S H, S^2 H\}.
  \]
\end{definition}

Note that $\omega \in \Ss$ and that $\Mm\Ss=\Ss\Mm$ is the 162-element
subgroup of $\Clifford$ which consists of generalized permutation
matrices. Note moreover that the syllables used in
\cref{def:canonical} are the elements of $\Ll\Mm T$.

\begin{lemma}
  The following relations hold.
  \begin{align}
    \Clifford&=\Ll\Mm\Ss\label{eq:CliffSub}\\ \Ss T&=T
    \Ss\label{eq:STrelation}\\ T T &= H^2 T H^2 Z\subseteq\Mm T \Mm
    \Ss\label{eq:TTrelation}\\ T H^2
    T&=H^2\subseteq\Mm.\label{eq:TH2Trelation}
  \end{align}
\end{lemma}

\begin{proof}
  \cref{eq:CliffSub} follows from the fact that the Clifford operators
  are a disjoint union of the cosets of $\Ss$ which are $\Ss$,
  $H^2\Ss$, $H\Ss$, $H^3\Ss$, $S H\Ss$, $S H^3\Ss$, $S^2 H\Ss$, and
  $S^2 H^3\Ss$. \cref{eq:STrelation} follows the three commutation
  relations $S T = T S$, $X T = T S X$, and $\omega T = T
  \omega$. Finally, \cref{eq:TTrelation,eq:TH2Trelation} follow from
  direct computation.
\end{proof}

\begin{lemma}
  \label{lem:TPow}
  An integer power of $T$ is either a Pauli operator or is Clifford
  equivalent to $T$. That is, for $a\in\Z$ we have
  \[
    T^a=\begin{cases} Z^\frac{a}{3} & a=0\mod 3\\ T Z^\frac{a-1}{3} &
    a=1\mod 3\\ H^2 T H^2 Z^\frac{a+1}{3} & a=2\mod 3
    \end{cases}
  \]
\end{lemma}

\begin{proof}
  This is a consequence of \cref{eq:TTrelation,eq:TH2Trelation} and
  the relations $T^9=\Id$ and $T Z=Z T$.
\end{proof}

\begin{definition}
  The \emph{Clifford-prefix} of a syllable $M=M'T\in\Ll\Mm T$ is the
  Clifford operator $M'\in\Ll\Mm$ that precedes $T$.
\end{definition}

\begin{proposition}
  \label{prop:canonicalexistence}
  Every Clifford+$T$ operator can be represented by a circuit in
  canonical form.
\end{proposition}

\begin{proof}
  We first show that if $U$ is a canonical form and $A$ is one of the
  generators of the Clifford+$T$ group then $UA$ admits a canonical
  form. In the case where $A$ is a Clifford operator there is nothing
  to show so we can assume that $A$ is a $T$ gate. We now proceed by
  induction on the $T$-count of $U$.
  \begin{itemize}
  \item If $U$ has $T$-count 0 then by
    \cref{eq:CliffSub,eq:STrelation,eq:canonicaldef} we have $U
    T\in\Ll\Mm\Ss T =\Ll\Mm T \Ss\subseteq\Ll\Mm T \Clifford$ so that
    $UA$ has a canonical form of $T$-count 1.
  \item If $U$ has $T$-count 1 then by
    \cref{eq:CliffSub,eq:canonicaldef} we know that $U\in\Ll \Mm T \Ll
    \Mm \Ss$. Using \cref{eq:STrelation,eq:TTrelation,eq:TH2Trelation}
    we get
    \begin{align*}
    U T& \in\Ll \Mm T \Ll \Mm \Ss T\\
    &= \Ll \Mm T \Ll \Mm T \Ss\\
    &= \Ll \Mm T \Ll' \Mm T \Ss \cup \Ll \Mm T \Mm T \Ss\\
    &= \Ll \Mm T \Ll' \Mm T \Ss \cup \Ll \Mm T T \Ss \cup \Ll \Mm T
    H^2 T \Ss\\
    &= \Ll \Mm T \Ll' \Mm T \Ss \cup \Ll \Mm H^2 T H^2 Z \Ss \cup \Ll
    \Mm H^2 \Ss\\
    &= \Ll \Mm T \Ll' \Mm T \Ss \cup \Ll \Mm T H^2 \Ss \cup \Ll \Mm
    \Ss\\
    &\subseteq \Ll \Mm T \Ll' \Mm T \Clifford \cup \Ll \Mm T \Clifford
    \cup \Clifford.
    \end{align*}
    It follows that $UA$ has a canonical form of $T$-count 0, 1, or 2.
  \item If $U$ has $T$-count $\ell>1$ then we can use
    \cref{eq:CliffSub,eq:canonicaldef} again to write $U$ as an
    element of $\Ll\Mm T (\Ll'\Mm T)^{\ell-2}\Ll'\Mm T\Ll \Mm \Ss$. We
    can now reason as in the previous case to show that
    \[
      U T\in\Ll \Mm T (\Ll' \Mm T)^\ell\Clifford \cup \Ll \Mm T (\Ll'
      \Mm T)^{\ell-1}\Clifford \cup \Ll \Mm T (\Ll' \Mm
      T)^{\ell-2}\Clifford.
    \]
    And it follows that $UA$ has a canonical form of $T$-count
    $\ell-1$, $\ell$, or $\ell+1$.
  \end{itemize}
  Now let $U$ be a canonical form and $A$ be either a Clifford
  operator or a power $T$. Assume moreover that the $T$-count of $U$
  is $\ell$ and the $T$-count of $A$ is $k$. Then the above argument,
  together with \cref{lem:TPow}, imply that $UA$ has a canonical form
  of $T$-count at most $k+\ell$.

  To complete the proof, let $V$ be a Clifford+$T$ operator. Then
  $V=A_1 \ldots A_n$ where every $A_i$ either a Clifford operator or a
  power of $T$. Starting with the identity operator, one may then
  proceed by rightward induction on $n$ to put $V$ in canonical form.
\end{proof}

\begin{corollary}
  \label{cor:canonicalalgorithm}  
  There exists an algorithm to rewrite any Clifford+$T$ circuit into
  canonical form. The algorithm runs in time linear in the gate-count
  of the input circuit.
\end{corollary}

\begin{proof}
  This is a consequence of the constructive proof of
  \cref{prop:canonicalexistence}. Indeed, a constant number of
  operations are needed to update the at most six of rightmost
  operators of a canonical form upon right-multiplication by a
  Clifford+$T$ operator. Any Clifford+$T$ operator of length $n$ can
  therefore be put in canonical form in $\Oo(n)$ steps.
\end{proof}

\begin{remark}
  \label{rmk:TMincanonical}
  Suppose that $V$ is a Clifford+$T$ circuit for some operator $U$ and
  that $V'$ is the canonical form for $U$ obtained by applying
  \cref{cor:canonicalalgorithm} to $V$. If $\ell$ is the $T$-count of
  $V$ and $\ell'$ is the $T$-count of $V'$ then $\ell' \leq
  \ell$. This follows from the fact that the algorithm of
  \cref{cor:canonicalalgorithm} never increases the $T$-count of a
  circuit.
\end{remark}

We close this section by discussing an alternative canonical form for
single-qutrit Clifford+$T$ circuits. The canonical form of
\cref{def:canonical} is inspired by the one introduced by Matsumoto
and Amano in \cite{MA08} for single-qubit Clifford+$T$ circuits. In
\cite{FGKM15}, Forest and others introduced a channel representation
for single qubit Clifford-cyclotomic circuits. When restricted to
single Clifford+$T$ circuits their channel representation can be
interpreted as a sequence of $\pi/4$ rotations about the $x$-, $y$-,
or $z$-axes of the Bloch Sphere (followed by a single Clifford
operator). This sequence is subject to the condition that consecutive
rotations revolve around different axes. Below, we define an analogue
for single-qutrit Clifford+$T$ circuits.

\begin{definition}
  \label{def:tpgates}
  Let $P$ be a Pauli operator and let $\lambda_0$, $\lambda_1$, and
  $\lambda_2$ be the following real numbers:
  \[
  \lambda_0 := \frac{1+ \zeta + \zeta^8}{3}, \quad \lambda_1 :=
  \frac{1+ \zeta^2 + \zeta^7}{3}, \quad \mbox{and} \quad \lambda_2 :=
  \frac{1+ \zeta^4 + \zeta^5}{3}.
  \]
  Then the \emph{$P$-axis $T$ gate} \emph{$T_P$} is defined as $T_P :=
  \lambda_0I + \lambda_1P + \lambda_2P^2$.
\end{definition}

\begin{definition}
  A \emph{channel form} is a single-qutrit Clifford+$T$ circuit of the
  form
  \begin{align}
  T_{P_1^{n_1}} T_{P_2^{n_2}}\ldots T_{P_\ell^{n_\ell}} C
  \end{align}
  where $\ell\in\N$, $P_i\in\s{Z,X,XZ,XZ^2}$,
  $n_i\in\Z_3\backslash\{1\}$, $P_i \neq P_{i+1}$ and $C\in\Clifford$.
\end{definition}

It can be shown that channel forms are in bijective correspondence
with the canonical forms of \cref{def:canonical} so that every
Clifford+$T$ operator admits a channel form. Moreover, this
correspondence preserves the $T$-count.

\section{Uniqueness of canonical forms}
\label{sec:uniqueness}

\cref{prop:canonicalexistence} showed that every Clifford+$T$ operator
can be represented by a circuit in canonical form. In this section, we
show that this representation is unique in the sense that if $M$ and
$N$ are different canonical forms that $M$ and $N$ represent different
Clifford+$T$ operators.

\subsection{Algebraic preliminaries}

\begin{definition}
  \label{def:rings}
  Let $\alpha = \sin(2\pi/9)$ and $\gamma=1-\zeta$, where $\zeta =
  e^{2\pi i/9}$ as in \cref{sec:ct}. We define six extensions of $\Z$.
\begin{itemize}
  \item $\Z[\omega]=\left\{a+b\omega\;\middle|\;a,b\in\Z\right\}$
  \item $\Z[\zeta]=\left\{a+b\zeta+c\zeta^2+d
    \zeta^3+e\zeta^4+f\zeta^5\;\middle|\;a, b,c,d,e,f\in\Z\right\}$
  \item $\Z[\frac{1}{\gamma}]= \left\{\frac{A}{\gamma^k}\;\middle|\;A
    \in\Z[\zeta],k\in\N\right\}$
  \item $\Z[\frac{1}{2}] =
    \left\{\frac{a}{2^k}\;\middle|\;a\in\Z,k\in\N\right\}$
  \item $\Z[\alpha]= \left\{A+B\alpha+C\alpha^2+D\alpha^3+E\alpha^4 +
    F\alpha^5\;\middle|\;A,B,C,D,E,F\in\Z[\frac{1}{2}]\right\}$
  \item $\Z[\frac{1}{2},\frac{1}{\alpha}]=
    \left\{\frac{A}{\alpha^k}\;\middle|\;A\in\Z[\alpha],
    k\in\N\right\}$
  \end{itemize}
  The ring $\Z[\omega]$ is known as the ring of \emph{cyclotomic
    integers of degree 3} while the ring $\Z[\zeta]$ is known as the
  ring of \emph{cyclotomic integers of degree 9}. The ring
  $\Z[\frac{1}{2}]$ is known as the ring of \emph{dyadic fractions}.
\end{definition}
    
\begin{remark}
  \label{rmk:AlphaRels}
  We record here some important relations involving $\alpha$ which
  will be useful in what follows.
  \begin{itemize}
  \item $\frac{\sqrt{3}}{2}=3\alpha-4\alpha^3$
  \item $\sqrt{3}=6\alpha-8\alpha^3$
  \item $3=36\alpha^2-96\alpha^4+64\alpha^6$
  \item $\alpha^6=\frac{3}{2^6}-
    \frac{9}{2^4}\alpha^2+\frac{3}{2}\alpha^6$
  \item $\alpha=\frac{\frac{3}{2^6}}{\alpha^5}-
    \frac{\frac{9}{2^4}}{\alpha^3}+\frac{\frac{3}{2}}{\alpha}$
  \item $\frac{1}{2}=-1+18\alpha^2-48\alpha^4+32\alpha^6$
  \end{itemize}
  In particular, the fifth relation implies that $\Z[\alpha]$ is a
  subring of $\Z\left[\frac{1}{2},\frac{1}{\alpha}\right]$.
\end{remark}

\begin{remark}
  The entries of any Clifford+$T$ operator belong to the ring
  $\Z[\frac{1}{\gamma}]$. To see this, note that it holds for the
  generators $S$, $H$, and $T$, since
  \[
  \frac{1}{\sqrt{-3}} = \frac{1}{\gamma^3} (-1+\zeta-\zeta^2-\zeta^3-
  2\zeta^4+2\zeta^5)\in\Z[\frac{1}{\gamma}].
  \]
\end{remark}

\begin{definition}[Residue]
  The \emph{residue map} $\rho$ is the ring homomorphism
  $\rho:\Z[\alpha]\rightarrow\Z_3$ defined by $\rho(q)=q\mod \alpha$.
\end{definition}

It follows from \cref{rmk:AlphaRels} that $\rho(\sqrt{3})=0$,
$\rho(3)=0$, and $\rho(\frac{1}{2})=2$. These equalities give an
intuition of how one might compute $\rho(q)$ given $q\in\Z[\alpha]$.
First write $q$ as a sum $q=c_0\alpha^0 + \ldots + c_5\alpha^5$ with
each $c_j\in\Z[\frac{1}{2}]$ such that $c_j=\frac{a_j}{2^{b_j}}$ where
$a_j\in\Z$ and $b_j\in\N$. Then $\rho(q)=\rho(c_0)=\rho(a_0 /
2^{b_0})=(2^{b_0}a_0)\mod 3.$

\begin{definition}[Denominator Exponent]
  For every $q\in\Z[\frac{1}{2},\frac{1}{\alpha}]$, there exists some
  $k\geq 0$ such that $\alpha^k q\in\Z[\alpha]$. Such a $k$ is called
  a \emph{denominator exponent} of $q$, and the least such $k$ is
  called the \emph{least denominator exponent} (LDE) of $q$. An
  integer $k\geq 0$ is a denominator exponent of a vector or matrix
  with entries in $\Z[\frac{1}{2},\frac{1}{\alpha}]$ if $k$ is a
  denominator exponent of every entry in the vector or matrix. The
  least such $k$ is the least denominator exponent of the vector or
  matrix.
\end{definition}

\begin{definition}[$k$-Residue]
  Let $q\in\Z[\frac{1}{2},\frac{1}{\alpha}]$ and let $k$ be a
  denominator exponent of $q$. Then the \emph{$k$-residue} of $q$,
  $\rho_k(q)$ and is defined as $\rho_k(q)=\rho(\alpha^k q) \in
  \Z_3$. The $k$-residue of a vector or matrix is defined
  component-wise.
\end{definition}	

\begin{lemma}
  \label{lem:LDEcond}
  Let $q\in\Z[\frac{1}{2},\frac{1}{\alpha}]$ and let $k\in\N$ be a
  denominator exponent of $q$. Then $k$ is the LDE of $q$ if and only
  if $\rho_k(q)\neq0\mod3$ or $k=0$.
\end{lemma}

\begin{proof}
  If $k=0$ then $k$ is a denominator exponent of $q$ if and only if it
  is the LDE of $k$. So suppose that $k>0$. Since $k$ is a denominator
  exponent of $q$ we can write $q$ as
  \[
  q=\frac{1}{\alpha^k}\sum_{j=0}^5 c_j\alpha^j
  \]
  with each $c_j\in\Z[\frac{1}{2}]$. Note that
  $\rho_k(q)=\rho(c_0)$. Since $k$ is not the LDE of $q$, we can
  rewrite $q$ as
  \[
  q=\frac{1}{\alpha^{k-1}}\left[\alpha^{-1}c_0+\sum_{j=0}^4
    c_{j+1}\alpha^j\right]
  \]
  where it must be the case that $\alpha^{-1}c_0\in\Z[\alpha]$. If
  $\rho(c_0)=0\mod3$, then we can write $c_0=3c_0'$ for some
  $c_0'\in\Z[\frac{1}{2}]$ and have
  \[
  \alpha^{-1}c_0=c_0'\frac{3}{\alpha} =
  c_0'\left(36\alpha-96\alpha^3+64\alpha^5\right)\in\Z[\alpha]
  \]
  where the term $\frac{3}{\alpha}$ is simplified using
  \cref{rmk:AlphaRels}. This proves the ``only if'' direction. On the
  other hand, if $\rho(c_0)=r\neq0\mod3$, then we have $c_0=r+3c_0''$
  for some $c_0''\in\Z[\frac{1}{2}]$ and $r\in\{1,2\}$ and can write
  \[
  \alpha^{-1}c_0 = \frac{r}{\alpha}+c_0''\frac{3}{\alpha} =
  \frac{r}{\alpha}+c_0'\left(36\alpha-96\alpha^3+64\alpha^5\right).
  \]
  The second term is in $\Z[\alpha]$, and so
  $\alpha^{-1}c_0\in\Z[\alpha]$ would only hold in this case if
  $\frac{r}{\alpha}$ is in $\Z[\alpha]$ as well. For $r\in\{1,2\}$
  this is not the case, leading to a contradiction and proving the
  ``if'' direction.
\end{proof}

\begin{remark}
  \label{rmk:ABrho}
  Let $A$ and $B$ be two matrices over
  $\Z[\frac{1}{2},\frac{1}{\alpha}]$ with LDE $k_A$ and $k_B$
  respectively. Then if $k\geq k_A$ and $k\geq k_B$ we have
  $\rho_k(A+B)=\rho_k(A)+\rho_k(B)$. Similarly, if $k_1\geq k_A$,
  $k_2\geq k_B$ and $k'=k_1+k_2$ such that $k_1\geq k_A$ and $k_2\geq
  k_B$ then $\rho_{k'}(AB) =
  \rho_{k_1}(A)\cdot\rho_{k_2}(B)$. Furthermore, if $A$ has the
  property that $k_A=0$ and that $\frac{1}{\alpha}A$ has entries in
  $\Z[\alpha]$, then
  \[
  \rho_{k'}(AB) =
  \rho_0\left(\frac{1}{\alpha}A\right)\cdot\rho_{k'+1}(B)
  \]
  for any $k'\geq k_B$. Likewise, if $k_B=0$ and $(1/\alpha)B$ has
  entries in $\Z[\alpha]$, then $\rho_{k'}(A B) =
  \rho_{k'+1}(A)\cdot\rho_0((1/\alpha) B)$. Finally, if $\ell >k_A$
  then $\rho_{\ell}(A)=0_{m\times n}$ by \cref{lem:LDEcond}.
\end{remark}

\subsection{The adjoint representation}

We will use an alternative representation for Clifford+$T$
operators. Let $P$ be a Pauli operator and define the operators $P_+$
and $P_-$ as
\[
  P_+:=\frac{1}{\sqrt{\tr\left[ \left(P+P^\dagger\right)^2\right]}}
  (P+P^\dagger) \quad \mbox{and} \quad P_-:=\frac{i}{\sqrt{-\tr\left[
        \left(P-P^\dagger\right)^2\right]}}(P-P^\dagger).
\]
Now consider the sets $\Qq=\{1/\sqrt{3},Z_+,X_+,(XZ)_+,
(XZ^2)_+,Z_-,X_-,(XZ)_-,(XZ^2)_-\}$ and $\Qq'=\Qq\backslash
\{1/\sqrt{3}\}$ where $1/\sqrt{3}=1/\sqrt{3}\Id$. These sets have a
number of properties:
\begin{itemize}
  \item The set $\Qq$ is a complete orthonormal basis for the set
    $\Matrices_3(\C)$ of $3\times3$ complex matrices with respect to
    the following inner product $\langle A,B\rangle =\langle
    B,A\rangle^* =\tr(A B^\dagger)$. That is, if $Q_i,Q_j\in\Qq$ then
    $\langle Q_i,Q_j\rangle=\delta_{i,j}$.
  \item Every $Q\in\Qq$ is Hermitian
  \item Every $Q\in\Qq'$ is traceless.
  \item Every $Q\in\Qq'$ is of one of two forms: either the matrix $Q$
    is of the type $P_+$ and is such that $P$ is a Pauli matrix with
    $P_+=\frac{1}{\sqrt{6}}(P+P^2)$ or it is of the type $P_-$ and is
    such that $P$ is a Pauli matrix with
    $P_-=\frac{i}{\sqrt{6}}(P-P^2)$.
\end{itemize}
Note that the set $\Qq$, much like the set of Gell-Mann matrices, does
\emph{not} form a group.

Because every element of $\Qq$ is Hermitian, any unit trace $3\times3$
Hermitian matrix $\rho$ may be written as
\begin{align}
  \label{eq:rhodef}
  \rho=\frac{1}{3}\Id+\frac{1}{\sqrt{3}}\sum_{Q_i\in\Qq'}c_i Q_i
\end{align}
with $c_i\in\R$. Conjugation by a unitary operator $U$ preserves the
trace of $\rho$. The action of $U$ will therefore send each $c_i$ to
some $c_i'\in\R$, since $\rho'=U\rho U^\dagger$ is still
Hermitian. This encourages us to define an adjoint representation for
unitary operators using $\Qq'$.

\begin{definition}
  Let $\hat{\rho}$ and $\hat{\rho}'$ be the eight-component vectors
  composed of the real $c_i$ and $c_i'$ as in \cref{eq:rhodef} for the
  $3\times3$ Hermitian matrices $\rho$ and $\rho'$. Then the adjoint
  representation of a unitary operator $U$, denoted $\hat{U}$, is
  defined by
  \[
  \hat{U}\hat{\rho}=\hat{\rho}'\iff U \rho U^\dagger = \rho'.
  \]
\end{definition}

\begin{remark}
  Composition of operators in the adjoint basis is equivalent to
  matrix multiplication, which can be seen as follows. If
  \[
  \hat{U}_1\hat{\rho}=\hat{\rho}'\iff U_1 \rho U_1^\dagger = \rho',
  \]
  then we have
  \[
  \hat{U}_2\hat{\rho}'=(\hat{U}_2\hat{U}_1)\hat{\rho}\iff U_2
  \rho'U_2^\dagger=U_2 U_1 \rho U_1^\dagger U_2^\dagger=(U_2
  U_1)\rho(U_2 U_1)^\dagger
  \]
\end{remark}

To maintain consistency, we impose an ordering for the $c_i$ mirroring
the ordering of
\[
\Qq'=\{Z_+,X_+,(XZ)_+,(XZ^2)_+,Z_-,X_-,(XZ)_-,(XZ^2)_-\}.
\]
This ordering allows us to write an explicit definition for $\hat{U}$
using our inner product.
\begin{lemma}
  Let $Q_i\in\Qq'$ with the ordering of $\Qq'$ fixed as above and
  $Q_i$ the $i$th element of $\Qq'$. The adjoint representation
  $\hat{U}$ of a Unitary operator $U$ may be calculated
  \[
  \hat{U}_{i,j}=\langle Q_i, U Q_j U^\dagger \rangle=\tr\left[Q_i U
    Q_j U^\dagger\right]
  \]
\end{lemma}

\begin{proof}
  This follows directly from the orthonormality and Hermiticity of the
  $Q_i$.
\end{proof}

\begin{lemma}
  Every adjoint representation $\hat{U}$ of a unitary operator $U$ is
  a real, special orthogonal matrix.
\end{lemma}

\begin{proof}
  Every element of $\hat{U}_{i,j}$ is real due to the properties of
  our inner product and the cyclic properties of the trace:
  \[
  \hat{U}_{i,j}^*=\langle Q_i, U Q_j U^\dagger \rangle^*=\langle U Q_j
  U^\dagger, Q_i \rangle=\tr\left[U Q_j U^\dagger Q_i
    \right]=\tr\left[Q_i U Q_j U^\dagger \right]=\langle Q_i, U Q_j
  U^\dagger \rangle=\hat{U}_{i,j}
  \]
  To show that every $\hat{U}$ is orthogonal, consider the inverse of
  $U$. As $U$ is unitary, the inverse of $U$ is $U^\dagger$ and thus
  the inverse of $\hat{U}$ must be
  $\hat{\left(U^\dagger\right)}$. This gives
  \[
  \hat{U}_{i,j}^{-1}=\hat{\left(U^\dagger\right)}_{i,j}=\langle Q_i,
  U^\dagger Q_j U \rangle = \tr\left[Q_i U^\dagger Q_j U \right] =
  \tr\left[Q_j U Q_i U^\dagger \right] = \langle Q_j, U Q_i U^\dagger
  \rangle = \hat{U}_{j,i} = \hat{U}^\Trans_{i,j}
  \]
  It just remains to show that $\det\hat{U}=1$. Consider again the
  matrix $U$. As it is unitary, we know there is a Hermitian operator
  $A$ such that $U=\exp(-i A)$. We begin by examining the Trotter
  decomposition of $U$:
  \[
  U=\lim_{N\rightarrow\infty}
  \left[\exp\left(\frac{-i}{N}A\right)\right]^N.
  \]
  This in turn implies that we have
  \[
  \hat{U}=\lim_{N\rightarrow\infty} \left[
    \hat{V}\left(\frac{A}{N}\right)\right]^N
  \]
  where we have defined $\hat{V}(M)$ as the adjoint representation for
  the operator $\exp(-iM)$. From this form, we calculate
  $\det\hat{U}$:
  \begin{align*}
  \det\hat{U}&=\lim_{N\rightarrow\infty}
  \left[\det\left(\hat{V}\left(\frac{A}{N}\right)\right)\right]^N
  \\ &= \lim_{N\rightarrow\infty} \exp\left[N\tr\left[\log
      \left(\hat{V}\left(\frac{A}{N}\right)\right)\right]\right]
  \end{align*}
  To lowest orders in $\frac{1}{N}$, we compute
  $\hat{V}\left(\frac{A}{N}\right)$ using the Hadamard lemma:
  \[
  \left(\hat{V}\left(\frac{A}{N}\right)\right)_{j,k} =
  \tr\left[Q_j\exp\left(\frac{-i}{N}A\right)
    Q_k\exp\left(\frac{i}{N}A\right)\right] = \tr\left[Q_j Q_k\right]
  -\frac{i}{N}\tr \left[Q_j\left[A,Q_k\right]\right]+
  \Oo\left(\left(\frac{1}{N}\right)^2\right)
  \]
  The first term here is simply the elemental form of the identity
  $\delta_{j,k}$ and the second term makes use of the standard
  definition of an algebra commutator, $[A,B]=AB-BA$. Upon calculating
  the trace of the matrix logarithm of
  $\hat{V}\left(\frac{A}{N}\right)$ we have
  \[
  \tr\left[\log\left(\hat{V}\left(\frac{A}{N} \right)\right) \right] =
  \frac{-i}{N}\sum_j\tr\left[Q_j \left[A,Q_j\right]\right] +
  \Oo\left(\left(\frac{1}{N}\right)^2\right) =
  \Oo\left(\left(\frac{1}{N}\right)^2\right)
  \]
  with the leading term disappearing due to the cyclic properties of
  the trace. Finally, this yields the desired result:
  \[
  \det\hat{U}=\lim_{N\rightarrow\infty}
  \exp\left[N\cdot\Oo\left(\left(\frac{1}{N}\right)^2\right)\right] =
  \lim_{N\rightarrow\infty}1+\Oo\left(\frac{1}{N}\right)=1
  \]
\end{proof}	

\begin{definition}[Quadrants]
  Let $\hat{U}$ be the adjoint representation of a unitary operator
  $U$. Let $\Qq'_+=\{Z_+,X_+,(XZ)_+,(XZ^2)_+\}$ and
  $\Qq'_-=\{Z_-,X_-,(XZ)_-,(XZ^2)_-\}$ be two sets with ordering as
  specified. Let $Q_{i,+}\in\Qq'_+$ be the $i$th element of $\Qq'_+$
  and $Q_{i,-}\in\Qq'_-$ be the $i$th element of $\Qq'_-$. Then we
  define the four $4\times4$ \emph{quadrants} of $\hat{U}$ (starting
  from the upper-left quadrant and going counter-clockwise) as
  follows:
  \begin{align*}
    \left(\hat{U}_{++}\right)_{i,j}&=\langle Q_{i,+}, U Q_{j,+}
    U^\dagger \rangle\\ \left(\hat{U}_{-+}\right)_{i,j}&=\langle
    Q_{i,-}, U Q_{j,+} U^\dagger
    \rangle\\ \left(\hat{U}_{--}\right)_{i,j}&=\langle Q_{i,-}, U
    Q_{j,-} U^\dagger
    \rangle\\ \left(\hat{U}_{+-}\right)_{i,j}&=\langle Q_{i,+}, U
    Q_{j,-} U^\dagger \rangle
  \end{align*}	
\end{definition}

\begin{lemma}
  Every adjoint representation $\hat{D}$ of a diagonal unitary
  operator $D$ is symplectic.
\end{lemma}

\begin{proof}
  As this result is not vital to the remainder of the work, proof is provided in Appendix~\ref{app:symplectic}.
\end{proof}

\subsection{Uniqueness}

\begin{remark}
  \label{rmk:CliffAdj}
  The entries of every adjoint representation $\hat{C}$ of a Clifford
  operator $C$ belong to the set
  $\{0,\pm1,\pm\frac{1}{2},\pm\frac{\sqrt{3}}{2}\}$. Moreover,
  $\hat{C}_{++}$ and $\hat{C}_{--}$ are both $4\times4$ generalized
  permutation matrices with the same underlying nonzero pattern and
  with entries in the set $\{0,\pm 1,\pm\frac{1}{2}\}$. On the other
  hand, $\hat{C}_{+-}$ and $\hat{C}_{-+}$ are less-than-full rank
  $4\times4$ matrices with the same nonzero pattern and at most one
  nonzero entry per row and column, with the entries belonging to the
  set $\{0,\pm\frac{\sqrt{3}}{2}\}$. These properties may be verified
  by enumeration of the 216 distinct adjoint representations of the
  Clifford operators, the full set of which we denote
  $\hat{\Clifford}$. Referencing \cref{rmk:AlphaRels} we can also
  immediately see that the entries of every
  $\hat{C}\in\hat{\Clifford}$ belong to the ring $\Z[\alpha]$.
\end{remark}

Writing the generators of $\CliffordT$ in our adjoint representation,
we have
\[
\begin{array}{c}
\hat{S}=\left(\begin{array}{cccc|cccc}
        1 & 0 & 0 & 0 & 0 & 0 & 0 & 0\\
        0 & 0 & 0 & 1 & 0 & 0 & 0 & 0\\
        0 & 1 & 0 & 0 & 0 & 0 & 0 & 0\\
        0 & 0 & 1 & 0 & 0 & 0 & 0 & 0\\
        \hline
        0 & 0 & 0 & 0 & 1 & 0 & 0 & 0\\
        0 & 0 & 0 & 0 & 0 & 0 & 0 & 1\\
        0 & 0 & 0 & 0 & 0 & 1 & 0 & 0\\
        0 & 0 & 0 & 0 & 0 & 0 & 1 & 0
\end{array}\right),\quad
\hat{H}=\left(\begin{array}{cccc|cccc}
        0 & 1 & 0 & 0 & 0 & 0 & 0 & 0\\
        1 & 0 & 0 & 0 & 0 & 0 & 0 & 0\\
        0 & 0 & 0 & -\frac{1}{2} & 0 & 0 & 0 & -\frac{\sqrt{3}}{2}\\
        0 & 0 & 1 & 0 & 0 & 0 & 0 & 0\\
        \hline
        0 & 0 & 0 & 0 & 0 & 1 & 0 & 0\\
        0 & 0 & 0 & 0 & -1 & 0 & 0 & 0\\
        0 & 0 & 0 & \frac{\sqrt{3}}{2} & 0 & 0 & 0 & -\frac{1}{2}\\
        0 & 0 & 0 & 0 & 0 & 0 & -1 & 0
\end{array}\right),\\ \\
\hat{T}=\left(\begin{array}{cccc|cccc}
        1 & 0 & 0 & 0 & 0 & 0 & 0 & 0\\
        0 & t_1 & t_1 & t_2 & 0 & t_3 & t_3 & t_4\\
        0 & t_2 & t_1 & t_1 & 0 & t_4 & t_3 & t_3\\
        0 & t_1 & t_2 & t_1 & 0 & t_3 & t_4 & t_3\\
        \hline
        0 & 0 & 0 & 0 & 1 & 0 & 0 & 0\\
        0 & -t_3 & -t_3 & -t_4 & 0 & t_1 & t_1 & t_2\\
        0 & -t_4 & -t_3 & -t_3 & 0 & t_2 & t_1 & t_1\\
        0 & -t_3 & -t_4 & -t_3 & 0 & t_1 & t_2 & t_1
\end{array}\right)
\end{array}
\]
where we have defined
\[
  \begin{array}{cc}
  t_1=\frac{1}{\alpha^2}\left(-\frac{1}{2^2}+2\alpha^2-
  2\alpha^4\right), & t_2=\frac{1}{\alpha^2}
  \left(\frac{1}{2^3}-\frac{3}{2}\alpha^2+2\alpha^4\right),
  \\ t_3=\frac{1}{\alpha^3}\left(-\frac{1}{2^4}+\frac{1}{2}\alpha^2-
  \alpha^4\right), &
  t_4=\frac{1}{\alpha^3}\left(\frac{1}{2^3}-\alpha^2+\alpha^4\right),
  \end{array}
\]
Note that the entries of any $\CliffordT$ operator belong to the ring
$\Z[\frac{1}{\gamma}]$. This is true due to the fact that it holds for
the generators $S$, $R$, and $T$, which is clear once one recognizes
that
\[
  \frac{1}{\sqrt{-3}} = \frac{1}{\gamma^3}
  \left[-1+\zeta-\zeta^2-\zeta^3-2\zeta^4
    +2\zeta^5\right]\in\Z[\frac{1}{\gamma}].
\]
Furthermore, the adjoint representation of any $\CliffordT$ operator
has entries in the ring $\Z[\frac{1}{2},\frac{1}{\alpha}]$, which
follows from \cref{rmk:CliffAdj} and the fact that the statement holds
for the remaining generator, $\hat{T}$.

\begin{remark}
  \label{rmk:CliffProp}
  Consider $\rho_0(\hat{S})$ and $\rho_0(\hat{H})$. Under the residue
  map, these Cliffords become the following:
  \[
    \rho_0(\hat{S})=\left(\begin{array}{cccc|cccc}
            1 & 0 & 0 & 0 & 0 & 0 & 0 & 0\\
            0 & 0 & 0 & 1 & 0 & 0 & 0 & 0\\
            0 & 1 & 0 & 0 & 0 & 0 & 0 & 0\\
            0 & 0 & 1 & 0 & 0 & 0 & 0 & 0\\
            \hline
            0 & 0 & 0 & 0 & 1 & 0 & 0 & 0\\
            0 & 0 & 0 & 0 & 0 & 0 & 0 & 1\\
            0 & 0 & 0 & 0 & 0 & 1 & 0 & 0\\
            0 & 0 & 0 & 0 & 0 & 0 & 1 & 0
    \end{array}\right)\quad\mbox{and}\quad
    \rho_0(\hat{H})=\left(\begin{array}{cccc|cccc}
            0 & 1 & 0 & 0 & 0 & 0 & 0 & 0\\
            1 & 0 & 0 & 0 & 0 & 0 & 0 & 0\\
            0 & 0 & 0 & 1 & 0 & 0 & 0 & 0\\
            0 & 0 & 1 & 0 & 0 & 0 & 0 & 0\\
            \hline
            0 & 0 & 0 & 0 & 0 & 1 & 0 & 0\\
            0 & 0 & 0 & 0 & 2 & 0 & 0 & 0\\
            0 & 0 & 0 & 0 & 0 & 0 & 0 & 1\\
            0 & 0 & 0 & 0 & 0 & 0 & 2 & 0
    \end{array}\right)
  \]
  Thus, \emph{any} adjoint representation $\hat{C}$ of a Clifford
  operator is such that the following hold:
  \begin{itemize}
    \item The LDE of $\hat{C}$ is zero
    \item $\rho_0(\hat{C})$ is a generalized permutation matrix with
      entries in $\Z_3$
    \item $\rho_0(\hat{C}_{++})$ is a true permutation matrix
    \item $\rho_0(\hat{C}_{++})$ is a generalized permutation matrix
      with entries in $\Z_3$
    \item $\rho_0(\hat{C}_{+-})=\rho_0(\hat{C}_{-+})=0_{4\times 4}$
    \item By \cref{rmk:AlphaRels,rmk:CliffAdj} we have
      $\rho_0\left(\frac{1}{\alpha}\hat{C}_{+-}\right) =
      \rho_0\left(\frac{1}{\alpha}\hat{C}_{-+}\right)=0_{4\times 4}$
  \end{itemize}
  In particular, we explicitly write out the nonzero quadrants of
  every Clifford in the set $\Ll\Mm$:
  \[
    \begin{array}{cc}
            \rho_0(\hat{\Id}_{++})=\rho_0(\hat{H}^2_{++})=\left(\begin{array}{cccc}
              1 & 0 & 0 & 0\\
              0 & 1 & 0 & 0\\
              0 & 0 & 1 & 0\\
              0 & 0 & 0 & 1
      \end{array}\right) & 
      \rho_0(\hat{H}_{++})=\rho_0(\hat{H}^3_{++})=\left(\begin{array}{cccc}
              0 & 1 & 0 & 0\\
              1 & 0 & 0 & 0\\
              0 & 0 & 0 & 1\\
              0 & 0 & 1 & 0
      \end{array}\right)\\
      \\
      \rho_0(\hat{S}\hat{H}_{++})=\rho_0(\hat{S}\hat{H}^3_{++})=\left(\begin{array}{cccc}
              0 & 1 & 0 & 0\\
              0 & 0 & 1 & 0\\
              1 & 0 & 0 & 0\\
              0 & 0 & 0 & 1
      \end{array}\right) & 
      \rho_0(\hat{S}^2\hat{H}_{++})=\rho_0(\hat{S}^2\hat{H}^3_{++})=\left(\begin{array}{cccc}
              0 & 1 & 0 & 0\\
              0 & 0 & 0 & 1\\
              0 & 0 & 1 & 0\\
              1 & 0 & 0 & 0
      \end{array}\right)\\
      \\
      \rho_0(\hat{\Id}_{--})=-\rho_0(\hat{H}^2_{--})=\left(\begin{array}{cccc}
      1 & 0 & 0 & 0\\
      0 & 1 & 0 & 0\\
      0 & 0 & 1 & 0\\
      0 & 0 & 0 & 1
      \end{array}\right) & 
      \rho_0(\hat{H}_{--})=-\rho_0(\hat{H}^3_{--})=\left(\begin{array}{cccc}
      0 & 1 & 0 & 0\\
      2 & 0 & 0 & 0\\
      0 & 0 & 0 & 1\\
      0 & 0 & 2 & 0
      \end{array}\right)\\
      \\
      \rho_0(\hat{S}\hat{H}_{--})=-\rho_0(\hat{S}\hat{H}^3_{--})=\left(\begin{array}{cccc}
      0 & 1 & 0 & 0\\
      0 & 0 & 2 & 0\\
      2 & 0 & 0 & 0\\
      0 & 0 & 0 & 1
      \end{array}\right) & 
      \rho_0(\hat{S}^2\hat{H}_{--})=-\rho_0(\hat{S}^2\hat{H}^3_{--})=\left(\begin{array}{cccc}
      0 & 1 & 0 & 0\\
      0 & 0 & 0 & 1\\
      0 & 0 & 2 & 0\\
      2 & 0 & 0 & 0
      \end{array}\right)
    \end{array}
  \]
\end{remark}

\begin{definition}[$k$-Adjoint Residue]
  Let $\hat{U}$ be an $8\times8$ matrix with entries in
  $\Z[\frac{1}{2},\frac{1}{\alpha}]$. Furthermore, let $\hat{U}_{++}$
  permit the denominator exponent $k$, $\hat{U}_{-+}$ and
  $\hat{U}_{+-}$ permit the denominator exponent $k+1$, and
  $\hat{U}_{--}$ permit the denominator exponent $k+2$. The we define
  the \emph{$k$-adjoint residue} of $\hat{U}$, in symbols
  $\hat{\rho}_k(\hat{U})$ and with $\hat{\rho}_k:\Matrices_8 \left(
  \Z[\frac{1}{2},\frac{1}{\alpha}]\right)
  \rightarrow\Matrices_8\left(\Z_3\right)$, as follows:
  \[
  \hat{\rho}_k(\hat{U})=\left(\begin{array}{cc}
  \rho_k(\hat{U}_{++}) & \rho_{k+1}(\hat{U}_{+-})\\
  \rho_{k+1}(\hat{U}_{-+}) & \rho_{k+2}(\hat{U}_{--})
  \end{array}\right).
  \]
  When we write $\hat{\rho}_k(\hat{U}_{\pm\pm})$, it is to be
  understood we want the array associated with function $\hat{\rho}_k$
  as it applies to the ($\pm\pm$) quadrant. This means that when we
  write e.g. $\hat{\rho}_k(\hat{U}_{+-})$, we really mean
  $\rho_{k+1}(\hat{U}_{+-})$
\end{definition}

\begin{remark}
  \label{rmk:CliffPerm}
  Let us briefly examine the consequences of left- or
  right-multiplication by a Clifford when considering only the
  $k$-adjoint residue of a matrix. In particular, let $\hat{U}$ be
  some adjoint representation of an operator where $\hat{U}$ has
  entries in $\Z[\frac{1}{2},\frac{1}{\alpha}]$ and is such that
  $\hat{\rho}_k(\hat{U})$ is well defined. Right multiplication of
  $\hat{U}$ by an adjoint representation $\hat{C}$ of a Clifford would
  yield
  \[
  \hat{U}\cdot\hat{C}=\left(\begin{array}{cc}
    \hat{U}_{++}\hat{C}_{++}+\hat{U}_{+-}\hat{C}_{-+} &
    \hat{U}_{++}\hat{C}_{+-}+\hat{U}_{+-}\hat{C}_{--}\\
    \hat{U}_{-+}\hat{C}_{++}+\hat{U}_{--}\hat{C}_{-+}
    & \hat{U}_{-+}\hat{C}_{+-}+\hat{U}_{--}\hat{C}_{--}
  \end{array}\right).
  \]
  Calculating the relevant $k$-residues of the resulting matrix, we
  have the following relations:
  \begin{align*}
    \rho_k((\hat{U}\cdot\hat{C})_{++})&=\rho_k(\hat{U}_{++})\cdot\rho_0(\hat{C}_{++})+\rho_{k+1}(\hat{U}_{+-})\cdot\rho_0\left(\frac{1}{\alpha}\hat{C}_{-+}\right)\\
    &=\rho_k(\hat{U}_{++})\cdot\rho_0(\hat{C}_{++})\\
    \rho_{k+1}((\hat{U}\cdot\hat{C})_{-+})&=\rho_{k+1}(\hat{U}_{-+})\cdot\rho_0(\hat{C}_{++})+\rho_{k+2}(\hat{U}_{--})\cdot\rho_0\left(\frac{1}{\alpha}\hat{C}_{-+}\right)\\
    &=\rho_{k+1}(\hat{U}_{-+})\cdot\rho_0(\hat{C}_{++})\\
    \rho_{k+1}((\hat{U}\cdot\hat{C})_{+-})&=\rho_{k+1}(\hat{U}_{++})\cdot\rho_0(\hat{C}_{+-})+\rho_{k+1}(\hat{U}_{+-})\cdot\rho_0\left(\hat{C}_{--}\right)\\
    &=\rho_{k+1}(\hat{U}_{+-})\cdot\rho_0(\hat{C}_{--})\\
    \rho_{k+2}((\hat{U}\cdot\hat{C})_{+-})&=\rho_{k+1}(\hat{U}_{-+})\cdot\rho_1\left(\hat{C}_{+-}\right)+\rho_{k+2}(\hat{U}_{--})\cdot\rho_0\left(\hat{C}_{--}\right)\\
    &=\rho_{k+2}(\hat{U}_{--})\cdot\rho_0(\hat{C}_{--})
  \end{align*}
  Left multiplication by a Clifford yields a similar set of relations:
  \begin{align*}
    \rho_k((\hat{C}\cdot\hat{U})_{++})&=\rho_0(\hat{C}_{++})\cdot\rho_k(\hat{U}_{++})+\rho_0\left(\frac{1}{\alpha}\hat{C}_{+-}\right)\cdot\rho_{k+1}(\hat{U}_{-+})\\
    &=\rho_0(\hat{C}_{++})\cdot\rho_k(\hat{U}_{++})\\
    \rho_{k+1}((\hat{C}\cdot\hat{U})_{-+})&=\rho_0(\hat{C}_{-+})\cdot\rho_{k+1}(\hat{U}_{++})+\rho_0(\hat{C}_{--})\cdot\rho_{k+1}(\hat{U}_{-+})\\
    &=\rho_0(\hat{C}_{--})\cdot\rho_{k+1}(\hat{U}_{-+})\\
    \rho_{k+1}((\hat{C}\cdot\hat{U})_{+-})&=\rho_0(\hat{C}_{++})\cdot\rho_{k+1}(\hat{U}_{+-})+\rho_0\left(\frac{1}{\alpha}\hat{C}_{+-}\right)\cdot\rho_{k+2}(\hat{U}_{--})\\
    &=\rho_0(\hat{C}_{++})\cdot\rho_{k+1}(\hat{U}_{+-})\\
    \rho_{k+2}((\hat{C}\cdot\hat{U})_{--})&=\rho_1\left(\hat{C}_{-+}\right)\cdot\rho_{k+1}(\hat{U}_{+-})+\rho_0(\hat{C}_{--})\cdot\rho_{k+2}(\hat{U}_{--})\\
    &=\rho_0(\hat{C}_{--})\cdot\rho_{k+2}(\hat{U}_{--})
  \end{align*}
  By these equations we immediately have for any adjoint
  representation $\hat{C}$ of a Clifford the following multiplicative
  rules for $\hat{\rho}_k$:
  \[
  \hat{\rho}_k(\hat{U}\cdot\hat{C}) =
  \hat{\rho}_k(\hat{U})\cdot\rho_0(\hat{C})\quad\mbox{and}\quad
  \hat{\rho}_k(\hat{C}\cdot\hat{U}) = \rho_0(\hat{C})\cdot
  \hat{\rho}_k(\hat{U})
  \]
  By \cref{rmk:CliffProp}, we know $\rho_0(\hat{C})$ is simply a
  generalized permutation matrix with $\rho_0(\hat{C}_{++})$ a true
  permutation matrix and $\rho_0(\hat{C}_{--})$ a generalized
  permutation matrix with the same nonzero pattern as
  $\rho_0(\hat{C}_{++})$. This means that (left-) right-multiplication
  by a Clifford simply corresponds to a permutation of (rows) columns,
  with the first 4 undergoing a true permutation and the last 4
  receiving the same underlying permutation with potential (row-)
  column-wide multiplicative factors applied.
\end{remark}

\begin{definition}[Clifford Equivalence]
  Let $\hat{U}$ and $\hat{V}$ be adjoint representations of
  $\CliffordT$ operators $U$ and $V$ such that there exists a Clifford
  operator $C$ with adjoint representation $\hat{C}$ where
  $\hat{U}\cdot\hat{C}=\hat{V}$. If $\hat{\rho}_k(\hat{U})$ is well
  defined, then we know that $\hat{\rho}_k(\hat{V})$ is also well
  defined and call these $k$-adjoint residues \emph{Clifford
    equivalent}, in symbols
  $\hat{\rho}_k(\hat{U})\sim_\Clifford\hat{\rho}_k(\hat{V})$, by which
  we mean
  $\hat{\rho}_k(\hat{U})\cdot\rho_0(\hat{C})=\hat{\rho}_k(\hat{V})$. We
  also extend this notion to quadrants, meaning that
  $\hat{\rho}_k(\hat{U}_{\pm\pm})
  \sim_\Clifford\hat{\rho}_k(\hat{V}_{\pm\pm})$ if and only if
  $\hat{\rho}_k(\hat{U})\sim_\Clifford\hat{\rho}_k(\hat{V})$ for the
  particular Clifford $\hat{C}$.
\end{definition}

\begin{proposition}
  \label{prop:canonicalPatt}
  Let $U\in\CliffordT$ be a canonical form, and $\hat{U}$ be the
  adjoint representation of $U$. Let $n$ be the $T$-count of $U$. Then
  the least denominator exponent of $\hat{U}_{++}$ is $k=2n$ and one
  of the following holds:
\end{proposition}

\begin{itemize}
  \item $n=0$ and $U$ is a Clifford operator.
  \item $n>0$ and one of 8 distinguishable cases holds for $\hat{\rho}_{2n}(\hat{U})$:\\
  \begin{align*}
   \hat{\rho}_{2n}(\hat{U}_{++})\sim_\Clifford\rho_0(\hat{M}_{++})\cdot\left(\begin{array}{cccc}
            0 & 0 & 0 & 0\\
            0 & 2 & 2 & 2\\
            0 & 2 & 2 & 2\\
            0 & 2 & 2 & 2
    \end{array}\right),\qquad \hat{\rho}_{2n}(\hat{U}_{-+})\sim_\Clifford\rho_0(\hat{M}_{--})\cdot\left(\begin{array}{cccc}
    0 & 0 & 0 & 0\\
    0 & 1 & 1 & 1\\
    0 & 1 & 1 & 1\\
    0 & 1 & 1 & 1
    \end{array}\right),
  \end{align*}
and the leftmost syllable is $MT$ with Clifford prefix $M\in\Ll\Mm$
    
\end{itemize}

\begin{proof}
  By direct computation, these statements hold for all canonical forms
  up to $T$-count three. In particular, enumerating these canonical
  forms gives the further condition that
  \[
  \hat{\rho}_{2n}(\hat{U}_{+-})=\hat{\rho}_{2n}(\hat{U}_{++}) \quad
  \mbox{and}\quad\hat{\rho}_{2n}(\hat{U}_{--})=\hat{\rho}_{2n}(\hat{U}_{-+})
  \]
  for all canonical forms of $T$-count $n=2$ and $n=3$ without a
  rightmost Clifford. Let $\hat{U}_{n,M_n}$ be an adjoint
  representation of a canonical form with $T$-count $n>2$ and leftmost
  syllable $M_n$ with Clifford prefix $M_n'\in\Ll'\Mm$. Let
  $\hat{U}_{n-1,M_{n-1}}=\hat{M}_n^\Trans\hat{U}_{n,M_n}$ such that it
  is also an adjoint representation of a canonical form with $T$-count
  $n-1>1$ and leftmost syllable $M_{n-1}$ with Clifford prefix
  $M_{n-1}'\in\Ll'\Mm$. Consider left-multiplication of
  $\hat{U}_{n,M_n}$ by $\hat{T}$:
  \[
          \hat{T}\hat{U}_{n,M_n}=\left(\begin{array}{cc}
                  \hat{T}_{++} & \hat{T}_{+-}\\
                  -\hat{T}_{+-} & \hat{T}_{++}
          \end{array}\right)\left(\begin{array}{cc}
                  (\hat{U}_{n,M_n})_{++} & (\hat{U}_{n,M_n})_{+-}\\
                  (\hat{U}_{n,M_n})_{-+} & (\hat{U}_{n,M_n})_{--}
          \end{array}\right).
  \]
  As $M_1$ is the leftmost syllable of $\hat{U}_{n,M_n}$, we can also
  rewrite some of its quadrants as
  \begin{align*}
          (\hat{U}_{n,M_n})_{-+}&=(\hat{M}_n)_{-+}(\hat{U}_{n-1,M_{n-1}})_{++}+(\hat{M}_n)_{--}(\hat{U}_{n-1,M_{n-1}})_{-+}\\
          (\hat{U}_{n,M_n})_{--}&=(\hat{M}_n)_{-+}(\hat{U}_{n-1,M_{n-1}})_{+-}+(\hat{M}_n)_{--}(\hat{U}_{n-1,M_{n-1}})_{--}
  \end{align*}
  Using these substitutions, we may write the following equations for
  the resulting quadrant matrices of $\hat{T}\hat{U}_{n,M_n}$:
  \begin{align*}
          (\hat{T}\hat{U}_{n,M_n})_{++}&=\hat{T}_{++}(\hat{U}_{n,M_n})_{++}+\hat{T}_{+-}(\hat{M}_n)_{-+}(\hat{U}_{n-1,M_{n-1}})_{++}+\hat{T}_{+-}(\hat{M}_n)_{--}(\hat{U}_{n-1,M_{n-1}})_{-+}\\
          (\hat{T}\hat{U}_{n,M_n})_{-+}&=-\hat{T}_{+-}(\hat{U}_{n,M_n})_{++}+\hat{T}_{++}(\hat{U}_{n,M_n})_{-+}\\
          (\hat{T}\hat{U}_{n,M_n})_{+-}&=\hat{T}_{++}(\hat{U}_{n,M_n})_{+-}+\hat{T}_{+-}(\hat{M}_n)_{-+}(\hat{U}_{n-1,M_{n-1}})_{+-}+\hat{T}_{+-}(\hat{M}_n)_{--}(\hat{U}_{n-1,M_{n-1}})_{--}\\
          (\hat{T}\hat{U}_{n,M_n})_{--}&=-\hat{T}_{+-}(\hat{U}_{n,M_n})_{+-}+\hat{T}_{++}(\hat{U}_{n,M_n})_{--}.
  \end{align*}
  Assume that $\hat{U}_{n,M_n,M_{n-1}}$ and $\hat{U}_{n-1,M_{n-1}}$
  have the following $2n$- and $2(n-1)$-adjoint residues,
  respectively:
  \begin{align*}
    \hat{\rho}_{2n}(\hat{U}_{n,M_n})&\sim_\Clifford\rho_0(\hat{M}_n')\cdot\left(\begin{array}{cccc|cccc}
            0 & 0 & 0 & 0 & 0 & 0 & 0 & 0\\
            0 & 2 & 2 & 2 & 0 & 2 & 2 & 2\\
            0 & 2 & 2 & 2 & 0 & 2 & 2 & 2\\
            0 & 2 & 2 & 2 & 0 & 2 & 2 & 2\\
            \hline
            0 & 0 & 0 & 0 & 0 & 0 & 0 & 0\\
            0 & 1 & 1 & 1 & 0 & 1 & 1 & 1\\
            0 & 1 & 1 & 1 & 0 & 1 & 1 & 1\\
            0 & 1 & 1 & 1 & 0 & 1 & 1 & 1
    \end{array}\right)\\
    \hat{\rho}_{2(n-1)}(\hat{U}_{n-1,M_{n-1}})&\sim_\Clifford\rho_0(\hat{M}_{n-1}')\cdot\left(\begin{array}{cccc|cccc}
            0 & 0 & 0 & 0 & 0 & 0 & 0 & 0\\
            0 & 2 & 2 & 2 & 0 & 2 & 2 & 2\\
            0 & 2 & 2 & 2 & 0 & 2 & 2 & 2\\
            0 & 2 & 2 & 2 & 0 & 2 & 2 & 2\\
            \hline
            0 & 0 & 0 & 0 & 0 & 0 & 0 & 0\\
            0 & 1 & 1 & 1 & 0 & 1 & 1 & 1\\
            0 & 1 & 1 & 1 & 0 & 1 & 1 & 1\\
            0 & 1 & 1 & 1 & 0 & 1 & 1 & 1
    \end{array}\right)
  \end{align*}
  Now, consider $\hat{\rho}_{2(n+1)}(\hat{T}\hat{U}_{n,M_n})$. Using
  \cref{rmk:ABrho,rmk:CliffPerm} and our equations for
  $(\hat{T}\hat{U}_{n,M_n})_{\pm\pm}$, we have
  \begin{align*}
          \hat{\rho}_{2(n+1)}((\hat{T}\hat{U}_{n,M_n})_{++})=&\rho_2(\hat{T}_{++})\cdot\rho_{2n}((\hat{U}_{n,M_n})_{++})\\
          &+\rho_4(\hat{T}_{+-}(\hat{M}_n)_{-+})\cdot\rho_{2(n-1)}((\hat{U}_{n-1,M_{n-1}})_{++})\\
          &+\rho_3(\hat{T}_{+-}(\hat{M}_n)_{--})\cdot\rho_{2n-1}((\hat{U}_{n-1,M_{n-1}})_{-+})\\
          \hat{\rho}_{2(n+1)}((\hat{T}\hat{U}_{n,M_n})_{-+})=&-\rho_3(\hat{T}_{+-})\cdot\rho_{2n}((\hat{U}_{n,M_n})_{++})\\
          &+\rho_2(\hat{T}_{++})\cdot\rho_{2n+1}((\hat{U}_{n,M_n})_{-+})\\
          \hat{\rho}_{2(n+1)}((\hat{T}\hat{U}_{n,M_n})_{+-})=&\rho_2(\hat{T}_{++})\cdot\rho_{2n+1}((\hat{U}_{n,M_n})_{+-})\\
          &+\rho_4(\hat{T}_{+-}(\hat{M}_n)_{-+})\cdot\rho_{2n-1}((\hat{U}_{n-1,M_{n-1}})_{+-})\\
          &+\rho_3(\hat{T}_{+-}(\hat{M}_n)_{--})\cdot\rho_{2n}((\hat{U}_{n-1,M_{n-1}})_{--})\\
          \hat{\rho}_{2(n+1)}((\hat{T}\hat{U}_{n,M_n})_{--})=&-\rho_3(\hat{T}_{+-})\cdot\rho_{2n+1}((\hat{U}_{n,M_n})_{+-})\\
          &+\rho_2(\hat{T}_{++})\cdot\rho_{2n+2}((\hat{U}_{n,M_n})_{--})
  \end{align*}
  Enumeration of the 6 possibilities for
  $\rho_4(\hat{T}_{+-}(\hat{M}_n)_{-+})$ yields
  $\rho_4(\hat{T}_{+-}(\hat{M}_n)_{-+})=0_{4\times4}$. Similarly,
  evaluation of the 36 distinct cases for
  $\rho_3(\hat{T}_{+-}(\hat{M}_n)_{--})\rho_{2n-1}((\hat{U}_{n-1,M_{n-1}})_{-+})$
  and
  $\rho_3(\hat{T}_{+-}(\hat{M}_n)_{--})\rho_{2n}((\hat{U}_{n-1,M_{n-1}})_{--})$
  yield
  \[
  \rho_3(\hat{T}_{+-}(\hat{M}_n)_{--}) \cdot
  \rho_{2n-1}((\hat{U}_{n-1,M_{n-1}})_{-+}) =
  \rho_3(\hat{T}_{+-}(\hat{M}_n)_{--})\cdot\rho_{2n}((\hat{U}_{n-1,M_{n-1}})_{--})=0_{4\times4}.
  \]
  Finally, the 6 options for both
  $\rho_2(\hat{T}_{++})\cdot\rho_{2n+1}((\hat{U}_{n,M_n})_{-+})$ and
  $\rho_2(\hat{T}_{++})\cdot\rho_{2n+2}((\hat{U}_{n,M_n})_{--})$ again
  yield
  \[
  \rho_2(\hat{T}_{++})\cdot\rho_{2n+1}((\hat{U}_{n,M_n})_{-+}) =
  \rho_2(\hat{T}_{++})\cdot\rho_{2n+2}((\hat{U}_{n,M_n})_{--})=0_{4\times4}.
  \]
  This leaves us a simplified set of equations
  \begin{align*}
          \hat{\rho}_{2(n+1)}((\hat{T}\hat{U}_{n,M_n})_{++})&=\rho_2(\hat{T}_{++})\cdot\rho_{2n}((\hat{U}_{n,M_n})_{++})\\
          \hat{\rho}_{2(n+1)}((\hat{T}\hat{U}_{n,M_n})_{-+})&=-\rho_3(\hat{T}_{+-})\cdot\rho_{2n}((\hat{U}_{n,M_n})_{++})\\
          \hat{\rho}_{2(n+1)}((\hat{T}\hat{U}_{n,M_n})_{+-})&=\rho_2(\hat{T}_{++})\cdot\rho_{2n+1}((\hat{U}_{n,M_n})_{+-})\\
          \hat{\rho}_{2(n+1)}((\hat{T}\hat{U}_{n,M_n})_{--})&=-\rho_3(\hat{T}_{+-})\cdot\rho_{2n+2}((\hat{U}_{n,M_n})_{+-})			
  \end{align*}
  Direct evaluation of the 6 options for each term yield only one
  possible resulting adjoint representation, summarized as follows
  $\rho_2(\hat{T}_{++})\cdot\rho_{2n}((\hat{U}_{n,M_n,M_{n-1}})_{++})$
  gives only one possible result:
  \[
    \hat{\rho}_{2(n+1)}(\hat{T}\hat{U}_{n,M_n})\sim_\Clifford\left(\begin{array}{cccc|cccc}
            0 & 0 & 0 & 0 & 0 & 0 & 0 & 0\\
            0 & 2 & 2 & 2 & 0 & 2 & 2 & 2\\
            0 & 2 & 2 & 2 & 0 & 2 & 2 & 2\\
            0 & 2 & 2 & 2 & 0 & 2 & 2 & 2\\
            \hline
            0 & 0 & 0 & 0 & 0 & 0 & 0 & 0\\
            0 & 1 & 1 & 1 & 0 & 1 & 1 & 1\\
            0 & 1 & 1 & 1 & 0 & 1 & 1 & 1\\
            0 & 1 & 1 & 1 & 0 & 1 & 1 & 1
    \end{array}\right).
  \]
  As $\hat{T}\hat{U}_{n,M_n,M_{n-1}}$ is itself an adjoint
  representation of a canonical form with $T$-count $n+1$ and leftmost
  syllable $T$, left-multiplication by any element $\hat{M}_{n+1}'$
  from the adjoint representation for the set $\Ll\Mm$ is also an
  adjoint representation of a canonical form with $T$-count $n+1$ and
  leftmost syllable $M_{n+1}$ with Clifford Prefix $M_{n+1}'$. Calling
  this new operator $\hat{U}_{n+1,M_{n+1}}$, we have
  \[
    \hat{\rho}_{2(n+1)}(\hat{U}_{n+1,M_{n+1}})\sim_\Clifford\rho_0(\hat{M}_{n+1}')\cdot\left(\begin{array}{cccc|cccc}
            0 & 0 & 0 & 0 & 0 & 0 & 0 & 0\\
            0 & 2 & 2 & 2 & 0 & 2 & 2 & 2\\
            0 & 2 & 2 & 2 & 0 & 2 & 2 & 2\\
            0 & 2 & 2 & 2 & 0 & 2 & 2 & 2\\
            \hline
            0 & 0 & 0 & 0 & 0 & 0 & 0 & 0\\
            0 & 1 & 1 & 1 & 0 & 1 & 1 & 1\\
            0 & 1 & 1 & 1 & 0 & 1 & 1 & 1\\
            0 & 1 & 1 & 1 & 0 & 1 & 1 & 1
    \end{array}\right).
  \]
  This particular pattern is then persistent under an inductive
  argument, given two consecutive $T$-counts possess the stated
  properties. Because all $T$-count 2 and 3 canonical forms obey the
  requisite requirements, we thus have that any canonical form
  $\hat{U}_{n,M_n}$ of $T$-count $n\geq2$ and leftmost syllable $M_n$
  with Clifford prefix $M_n'\in\Ll\Mm$ will obey the relations
  \begin{align*}
  	\hat{\rho}_{2n}(\hat{U}_{++})\sim_\Clifford\rho_0(\hat{(M_n')}_{++})\cdot\left(\begin{array}{cccc}
  		0 & 0 & 0 & 0\\
  		0 & 2 & 2 & 2\\
  		0 & 2 & 2 & 2\\
  		0 & 2 & 2 & 2
  	\end{array}\right),\qquad \hat{\rho}_{2n}(\hat{U}_{-+})\sim_\Clifford\rho_0(\hat{(M_n')}_{--})\cdot\left(\begin{array}{cccc}
  		0 & 0 & 0 & 0\\
  		0 & 1 & 1 & 1\\
  		0 & 1 & 1 & 1\\
  		0 & 1 & 1 & 1
  	\end{array}\right).
  \end{align*}
  Enumeration of canonical forms of $T$-count one shows that they likewise have this property, and so coupled with the fact that the LDE
  of any Clifford operator is zero we have shown
  \cref{prop:canonicalPatt} to be true.
\end{proof}

\begin{proposition}
  \label{prop:canonicalunique}
  If $M$ and $N$ are different canonical forms then they represent
  different operators.
\end{proposition}

\begin{proof}
  \label{proof:canonicalunique}  
  Let $U$ and $V$ be different canonical forms with $T$-counts $n$ and
  $m$ respectively. If $n\neq m$, then by \cref{prop:canonicalPatt}
  the LDEs of $\hat{U}_{++}$ and $\hat{V}_{++}$ differ and so must $U$
  and $V$. This leaves the case when $n=m$. Let $U$ and $V$ differ
  such that their first mismatched syllable starting from the left is
  the $p$th syllable counting from the right, with $U_{p,0}$ the
  associated canonical form of $U$ truncated at this syllable as
  starting from the right such that $U=U_{n,p+1} U_{p,0}$. Then
  $\hat{U}_{n,p+1}^\Trans \cdot
  \hat{U}\neq\hat{U}_{n,p+1}^\Trans\cdot\hat{V}$ by
  \cref{prop:canonicalPatt}, and thus $U$ and $V$ are different. Now,
  let $U$ and $V$ be such that every syllable is identical, but their
  rightmost Cliffords are different. Then $U^\dagger
  V\in\Clifford\backslash\{\Id\}$ and therefore $U\neq V$. This
  enumerates all possible cases.
\end{proof}

\begin{corollary}
  \label{cor:canonicalAlg}  
  Let $U\in\mathcal{C}+T$ have adjoint representation $\hat{U}$ with
  LDE $k$ of $\hat{U}_{++}$. Then the canonical form $M$ associated
  with $U$ has $T$-count $n=\frac{k}{2}$ and can be efficiently
  computed in $\mathcal{O}(n)$ arithmetic operations.
\end{corollary}

\begin{proof}
  From $U$, we compute the adjoint representation $\hat{U}$ using a
  constant number of operations, in the process determining the LDE
  $k$ of $\hat{U}_{++}$. By \cref{prop:canonicalPatt}, we have two
  cases depending on the value of $k$. If $k=0$, $U$ is equivalent to
  a Clifford operator $C$ and $M$ can be found via lookup table. If
  $k>0$, $k$ is even by \cref{prop:canonicalPatt} and so let
  $n=\frac{k}{2}$. Then we can find the leftmost syllable $M_{n}$ in a
  constant number of operations by evaluating
  $\hat{\rho}_{2n}(\hat{U}_{++})$ and
  $\hat{\rho}_{2n}(\hat{U}_{-+})$. Now, calculate
  $\hat{U}'=\hat{M}_n^\Trans\hat{U}$ - by \cref{prop:canonicalPatt},
  we know $\hat{U}'$ is the adjoint representation of a canonical form
  with LDE $k-2$ of $\hat{U}_{++}'$ such that the $T$-count of
  $\hat{U}'$ is $n-1$. Carrying out this procedure recursively, we are
  left with the $U$ equivalent canonical form
  \[
  M=M_n M_{n-1}\ldots M_1 C
  \]
  where it took a constant number of operations to calculate each
  $M_i$ and $C$, thus requiring an overall runtime of $\Oo(n)$.
\end{proof}

We conclude with two important consequences of the uniqueness of
canonical forms.

\begin{proposition}
  \label{prop:TOpt}
  Canonical forms are \emph{$T$-optimal}: for any canonical form with
  $T$-count $n$ there are no equivalent $\CliffordT$ operators with a
  number of power of $T$ gates less than $n$.
\end{proposition}

\begin{proof}
  By \cref{prop:canonicalexistence} we know that every $\CliffordT$
  operator admits a canonical form, and by
  \cref{prop:canonicalunique}, we will have that these canonical forms
  are both unique. Furthermore, by \cref{rmk:TMincanonical}, we know
  that in putting any $\CliffordT$ operator into either canonical form
  by the algorithms laid out in \cref{cor:canonicalalgorithm}, the
  $T$-count may \emph{only} decrease compared to the number of power
  of $T$ gates. In combination, these statements suffice to show
  $T$-optimality.
\end{proof}

\begin{proposition}
  \label{prop:Tnum}
  Let $\epsilon>0$. There exists $U\in\mbox{SU(3)}$ whose
  $\epsilon$-approximation by a Clifford+$T$ circuit requires a number
  $n$ of $T$ gates where $n\gtrsim
  8\log_6\left(\frac{1}{\epsilon}\right)-K$ for $K\approx0.543$.
\end{proposition}

\begin{proof}
  This follows from a volume-counting argument. Indeed, there are
  $216/5(8\cdot6^n-3)$ canonical forms of $T$-count at most
  $n$. Moreover, each $\epsilon$-ball occupies a volume of
  $(\pi^4/24)\epsilon^8$ as $\epsilon$ asymptotes towards zero (by
  which the 8-dimensional manifold $SU(3)$ becomes locally
  Euclidean). We need to cover the full volume $\sqrt{3}\pi^5$ of
  $SU(3)$ to guarantee that every operator can be approximated up to
  $\epsilon$. Therefore $n$ needs to satisfy
  \[
  \frac{216}{5}(8\cdot6^n-3)\frac{\pi^4}{24}{\epsilon^8}
  \gtrsim\sqrt{3}\pi^5.
  \]
  From which the result follows.
\end{proof}

\section{Conclusion}
\label{sec:conc}

Significant advances in our understanding of the Clifford+$T$ group
for both single- and mutli-qubit circuits have been made in the past
decade. Analogous results for qudits of higher dimension, however,
remained elusive. In this paper we contribute to the theory of
single-qutrit Clifford+$T$ circuits by providing a canonical form for
single-qutrit Clifford+$T$ circuits. We show that every Clifford+$T$
operator admits a unique canonical representation and that this
representation is $T$-optimal. We leave the question of generalizing
these canonical forms to qudits of higher prime dimensions as an
avenue for future work.

\section{Acknowledgements}
\label{sec:acknowledgements}

NJR thanks Peter Selinger for valuable insights, and Alex Bocharov and
Vadym Kliuchnikov for stimulating discussions. ANG thanks Mark Howard
for providing the appropriate terminology for some of the techniques
used in the paper. ANG and NJR thank Earl Campbell for helpful
discussions. ANG, NJR, and JMT were partially funded by the Department
of Defense. This work was funded in part by the Army Research
Laboratory Center for Distributed Quantum Information.

\bibliographystyle{abbrv} \bibliography{qutrit}

\begin{thebibliography}{10}

\bibitem{PSSMBC}
V.~V. Albert, K.~Noh, K.~Duivenvoorden, D.~Young, R.~Brierley, P.~Reinhold,
  C.~Vuillot, L.~Li, C.~Shen, S.~Girvin, B.~Terhal, and L.~Jiang.
\newblock Performance and structure of single-mode bosonic codes.
\newblock aug 2017.

\bibitem{BBG2014}
A.~Blass, A.~Bocharov, and Y.~Gurevich.
\newblock Optimal ancilla-free {Pauli}+{$V$} circuits for axial rotations.
\newblock Dec. 2014.

\bibitem{B16}
A.~Bocharov.
\newblock A note on optimality of quantum circuits over metaplectic basis.
\newblock June 2016.

\bibitem{BCRS15}
A.~Bocharov, S.~X. Cui, M.~Roetteler, and K.~M. Svore.
\newblock Improved quantum ternary arithmetics.
\newblock Dec. 2015.

\bibitem{BCKZ15}
A.~Bocharov, X.~Cui, V.~Kliuchnikov, and Z.~Wang.
\newblock Efficient topological compilation for weakly-integral anyon model.
\newblock Apr. 2015.

\bibitem{BGS2013}
A.~Bocharov, Y.~Gurevich, and K.~M. Svore.
\newblock Efficient decomposition of single-qubit gates into {$V$} basis
  circuits.
\newblock {\em Phys. Rev. A}, 88:012313 (13 pages), 2013.

\bibitem{BRS16}
A.~Bocharov, M.~Roetteler, and K.~M. Svore.
\newblock Factoring with qutrits: Shor's algorithm on ternary and metaplectic
  quantum architectures.
\newblock May 2016.

\bibitem{Dawson-Nielsen}
C.~M. Dawson and M.~A. Nielsen.
\newblock The {Solovay-Kitaev} algorithm.
\newblock {\em Quantum Information and Computation}, 6(1):81--95, Jan. 2006.

\bibitem{FGKM15}
S.~Forest, D.~Gosset, V.~Kliuchnikov, and D.~McKinnon.
\newblock Exact synthesis of single-qubit unitaries over {Clifford}-cyclotomic
  gate sets.
\newblock Jan. 2015.

\bibitem{ma-remarks}
B.~Giles and P.~Selinger.
\newblock Remarks on {Matsumoto} and {Amano}'s normal form for single-qubit
  {Clifford+$T$} operators.
\newblock Dec. 2013.

\bibitem{G98}
D.~Gottesman.
\newblock Fault-tolerant quantum computation with higher-dimensional systems.
\newblock In {\em Selected Papers from the First NASA International Conference
  on Quantum Computing and Quantum Communications}, QCQC '98, pages 302--313,
  1998.

\bibitem{SSCOSADMA}
E.~Hostens, J.~Dehaene, and D.~M. B.
\newblock Stabilizer states and {C}lifford operations for systems of arbitrary
  dimensions, and modular arithmetic.
\newblock aug 2004.

\bibitem{HV12}
M.~Howard and J.~Vala.
\newblock Qudit versions of the qubit $\ensuremath{\pi}/8$ gate.
\newblock {\em Phys. Rev. A}, 86:022316, Aug 2012.

\bibitem{KSV2002}
A.~Y. Kitaev, A.~H. Shen, and M.~N. Vyalyi.
\newblock {\em Classical and Quantum Computation}.
\newblock Graduate Studies in Mathematics 47. American Mathematical Society,
  2002.

\bibitem{KBRY2015}
V.~Kliuchnikov, A.~Bocharov, M.~Roetteler, and J.~Yard.
\newblock A framework for approximating qubit unitaries.
\newblock Oct. 2015.

\bibitem{KBS2013}
V.~Kliuchnikov, A.~Bocharov, and K.~M. Svore.
\newblock Asymptotically optimal topological quantum compiling.
\newblock Oct. 2013.

\bibitem{KMM-practical}
V.~Kliuchnikov, D.~Maslov, and M.~Mosca.
\newblock Practical approximation of single-qubit unitaries by single-qubit
  quantum {Clifford} and {$T$} circuits.
\newblock Dec. 2012.

\bibitem{KMM-approx}
V.~Kliuchnikov, D.~Maslov, and M.~Mosca.
\newblock Asymptotically optimal approximation of single qubit unitaries by
  {Clifford} and {$T$} circuits using a constant number of ancillary qubits.
\newblock {\em Phys. Rev. Lett.}, 110:190502 (5 pages), 2013.

\bibitem{KMM}
V.~Kliuchnikov, D.~Maslov, and M.~Mosca.
\newblock Fast and efficient exact synthesis of single-qubit unitaries
  generated by {C}lifford and {T} gates.
\newblock {\em Quantum Info. Comput.}, 13(7-8):607--630, July 2013.

\bibitem{K96b}
E.~Knill.
\newblock Group representations, error bases and quantum codes.
\newblock Technical report, Los Alamos National Laboratory, 1996.
\newblock quant-ph/9608049.

\bibitem{K96a}
E.~Knill.
\newblock Non-binary unitary error bases and quantum codes.
\newblock Technical report, Los Alamos National Laboratory, 1996.
\newblock quant-ph/9608048.

\bibitem{MA08}
K.~Matsumoto and K.~Amano.
\newblock Representation of quantum circuits with {Clifford} and $\pi$/8 gates.
\newblock June 2008.

\bibitem{NCQEC}
M.~H. Michael, M.~Silveri, R.~T. Brierley, V.~V. Albert, J.~Salmilehto,
  L.~Jiang, and S.~M. Girvin.
\newblock New class of quantum error-correcting codes for a bosonic mode.
\newblock {\em Phys. Rev. X}, 6:031006, Jul 2016.

\bibitem{HEBQEC}
M.~Niu, I.~Chuang, and J.~Shapiro.
\newblock Hardware-efficient bosonic quantum error-correcting codes based on
  symmetry operators.
\newblock sep 2017.

\bibitem{PJKS}
S.~Prakash, A.~Jain, B.~Kapur, and S.~Seth.
\newblock A normal form for single-qutrit clifford+t operators.
\newblock mar 2017.

\bibitem{vsynth}
N.~J. Ross.
\newblock Optimal ancilla-free {Clifford}+{$V$} approximation of
  {$z$}-rotations.
\newblock {\em Quantum Information and Computation}, 15(11--12):932--950, 2015.

\bibitem{RS16}
N.~J. Ross and P.~Selinger.
\newblock Optimal ancilla-free {C}lifford+{T} approximation of
  \emph{z}-rotations.
\newblock {\em Quantum Information {\&} Computation}, 16(11{\&}12):901--953,
  2016.

\bibitem{S14}
P.~Selinger.
\newblock Efficient {Clifford}+{$T$} approximation of single-qubit operators.
\newblock {\em Quantum Information and Computation}, 2014.

\bibitem{WCAB15}
F.~H.~E. Watson, E.~T. Campbell, H.~Anwar, and D.~E. Browne.
\newblock Qudit color codes and gauge color codes in all spatial dimensions.
\newblock {\em Phys. Rev. A}, 92:022312, Aug 2015.

\end{thebibliography}
\newpage
\begin{appendices}
\section{}
\label{app:symplectic}
As $D$ is diagonal and unitary, we may write it as
\[
D=\left(\begin{array}{ccc} e^{i \beta_1} & 0 & 0\\ 0 & e^{i \beta_2}
& 0\\ 0 & 0 & e^{i \beta_3}
\end{array}\right)
\]
Direct computation of $\hat{D}$ yields
\[
\hat{D}=\left(\begin{array}{cccc|cccc} 1 & 0 & 0 & 0 & 0 & 0 & 0 &
0\\ 0 & d_1 & d_2 & d_3 & 0 & d_4 & d_5 & d_6\\ 0 & d_3 & d_1 &
d_2 & 0 & d_6 & d_4 & d_5\\ 0 & d_2 & d_3 & d_1 & 0 & d_5 & d_6 &
d_4\\ \hline 0 & 0 & 0 & 0 & 1 & 0 & 0 & 0\\ 0 & -d_4 & -d_5 &
-d_6 & 0 & d_1 & d_2 & d_3\\ 0 & -d_6 & -d_4 & -d_5 & 0 & d_3 &
d_1 & d_2\\ 0 & -d_5 & -d_6 & -d_4 & 0 & d_2 & d_3 & d_1
\end{array}\right)
\]
where we have defined
\begin{align*}
	d_1&=\frac{1}{3}\left[\cos(\beta_1-\beta_2)+\cos(\beta_2-\beta_3) +
	\cos(\beta_3-\beta_1)\right]\\ d_2&=\frac{1}{6}
	\left[2\cos(\beta_1-\beta_2) -
	\cos(\beta_2-\beta_3)-\cos(\beta_3-\beta_1) +
	\sqrt{3}\left(\sin(\beta_2-\beta_3) -
	\sin(\beta_3-\beta_1)\right)\right]\\ d_3 & =
	\frac{1}{6}\left[2\cos(\beta_1-\beta_2)-\cos(\beta_2-\beta_3) -
	\cos(\beta_3-\beta_1)-\sqrt{3}\left(\sin(\beta_2 - \beta_3) -
	\sin(\beta_3 - \beta_1)\right)\right]\\ d_4&=\frac{1}{3}
	\left[\sin(\beta_1-\beta_2) + \sin(\beta_2-\beta_3) +
	\sin(\beta_3-\beta_1)\right] \\ d_5 &
	=\frac{1}{6}\left[2\sin(\beta_1-\beta_2) - \sin(\beta_2-\beta_3) -
	\sin(\beta_3-\beta_1) - \sqrt{3}\left(\cos(\beta_2-\beta_3) -
	\cos(\beta_3-\beta_1)\right)\right] \\ d_6 & =
	\frac{1}{6}\left[2\sin(\beta_1-\beta_2) -
	\sin(\beta_2-\beta_3)-\sin(\beta_3 - \beta_1) +
	\sqrt{3}\left(\cos(\beta_2 -
	\beta_3)-\cos(\beta_3-\beta_1)\right)\right]
\end{align*}
This means we have $\hat{D}_{++}=\hat{D}_{--}=A$ and
$\hat{D}_{-+}=-\hat{D}_{+-}=-B$, subject to the conditions that $A
A^\Trans+B B^\Trans=A^\Trans A+B^\Trans B=\Id$, $A B^\Trans=B
A^\Trans$, and $A^\Trans B = B^\Trans A$ due to $\hat{D}$ being
special orthogonal. To see that these conditions suffice to show
$\hat{D}$ is symplectic, we must show
$\hat{D}^\Trans\Omega\hat{D}=\Omega$ with
\[
\Omega=\left(\begin{array}{cc}
0 & \Id_{4\times4}\\
-\Id_{4\times4} & 0
\end{array}\right).
\]
Using our properties for $A$ and $B$, we see
\[
\left(\begin{array}{cc}
A^\Trans & -B^\Trans\\
B^\Trans & A^\Trans
\end{array}\right)
\left(\begin{array}{cc}
0 & \Id\\
-\Id & 0
\end{array}\right)
\left(\begin{array}{cc}
A & B\\
-B & A
\end{array}\right)=
\left(\begin{array}{cc} B^\Trans A - A^\Trans B & B^\Trans B +
A^\Trans A\\ -A^\Trans A - B^\Trans B & -A^\Trans B + B^\Trans A
\end{array}\right)=
\left(\begin{array}{cc}
0 & \Id\\
-\Id & 0
\end{array}\right)
\]
and so $\hat{D}$ is symplectic.

\end{appendices}

\end{document}